\documentclass[12pt]{iopart}

 \usepackage{subfig}
\usepackage{graphicx,bbm,amssymb,amsopn}

\newcommand{\bra}[1]{{\langle #1 \vert}}

\newcommand{\ket}[1]{{\vert #1 \rangle}}

\newcommand{\ave}[1]{{\langle #1\rangle}}

\newcommand{\ii}{ {\rm i} }
\newcommand{\dd}{ {\rm d} }
\newcommand{\ZZ}{\mathbb{Z}}

\newcommand{\CC}{\mathbb{C}}

\newcommand{\y}{{\rm y}}
\newcommand{\x}{{\rm x}}
\newcommand{\z}{{\rm z}}
\newcommand{\LL}{{\hat {\cal L}}}

\def\tr{{\,{\rm tr}\,}}

\newtheorem{theorem}{Theorem}[section]

\newenvironment{proof}[1][Proof]{\begin{trivlist}
\item[\hskip \labelsep {\bfseries #1}]}{\end{trivlist}}

\newcommand{\qed}{\nobreak \ifvmode \relax \else
      \ifdim\lastskip<1.5em \hskip-\lastskip
      \hskip1.5em plus0em minus0.5em \fi \nobreak
      \vrule height0.75em width0.5em depth0.25em\fi}

\begin{document}

\title[A note on symmetry reductions of the Lindblad equation]{A note on symmetry reductions of the Lindblad equation: transport in constrained open spin chains}
\author{Berislav Bu\v ca and Toma\v{z} Prosen}

\address{Department of physics, Faculty of Mathematics and Physics, University of Ljubljana,
Jadranska 19, SI-1000 Ljubljana, Slovenia}

\date{\today}


\begin{abstract}
We study quantum transport properties of an open Heisenberg XXZ spin 1/2 chain driven by a pair of Lindblad jump operators satisfying a global `microcanonical' constraint, i.e. conserving the total magnetization. We will show that this system has an additional discrete symmetry which is particular to the Liouvillean description of the problem.
Such symmetry reduces the dynamics even more than what would be expected in the standard Hilbert space formalism and establishes existence of multiple steady states. Interestingly, numerical simulations of the XXZ model suggest that a pair of distinct non-equilibrium steady states becomes indistinguishable in the thermodynamic limit, and exhibit sub-diffusive spin transport in the easy-axis regime of anisotropy $\Delta > 1$.
\end{abstract}

\pacs{03.65.Yz, 03.65.Fd, 05.30.Fk, 75.10.Pq}

\maketitle

\section{Introduction}

Simulating interacting quantum many-body systems in-, and in particular, out-of- equilibrium is at the forefront of the current experimental and theoretical research (see e.g. \cite{Bloch}, \cite{Polkovnikov}).
While it has been recognized that treating large isolated quantum-many body systems (with generic physical properties) poses a formidable task to any method of theoretical description or simulation, it has soon become apparent -- first within the community of quantum optics \cite{Zoller} -- that often a macroscopic number of observables can be treated as quantum noise, and that the behavior of a few
essential physical observables can be extracted in terms of the formalism of open quantum systems \cite{Petruccione,Alicki}. From mathematical-physics point of view most notable is perhaps the
discovery of Lindblad, Gorini, Kossakowski and Sudarshan \cite{Lindblad,Gorini}, that within the Markovian approximation any time-development of an open system's (reduced) density operator $\rho(t)$, acting on a 
$N-$dimensional Hilbert space ${\cal H}$ describing the (slow) degrees of freedom of interest, can be described by the quantum Liouville equation in the form
\begin{equation}
\dot\rho(t)= \LL \rho(t) \equiv -\ii 
[H,\rho(t)] +\sum_{m = 1}^{N^2-1}\Bigl(L_m \rho(t) L_m^\dagger -\frac{1}{2}\{ L_m^\dagger L_m,\rho(t)\}\Bigr) \label{lind}
\end{equation}
where $[A,B]\equiv AB-BA$, $\{A,B\}\equiv  AB+BA$, $H$ is the Hermitian Hamiltonian operator, and $L_m$ is a set of at most $N^2-1$ so-called Lindblad operators, all over ${\cal H}$. We use units in which $\hbar=1$.
$\LL$ can be understood as a (super)operator acting on the space ${\cal B}(\cal H)$ of bounded operators over Hilbert space $\cal H$. 
The space ${\cal B}({\cal H})$ shall also be considered as a Hilbert space itself, equipped with the Hilbert-Schmidt inner product.
The Liouvillian flow (\ref{lind}) is the most general one which is (i)
local-in-time, (ii) respects the positivity $\rho(t) \ge 0$, and (iii) conserves the trace $\tr \rho(t) \equiv 1$, and forms the so-called quantum dynamical semi-group. However, the approach of open-quantum-system has only quite recently been applied to the non-equilibrium problems of condensed matter physics  \cite{Saito:03,Michel:03,Carlos,Steinigeweg:06,Wichterich:07,NJP:08,Michel:08,Yan:08}, such as the controversial quantum transport problem in one-dimension \cite{ZotosReview,Affleck,Fabian}.

The central object of concern in physics applications is the fixed point of dynamical semigroup, $\rho_{\infty} = \lim_{t\to\infty}\rho(t)$, or the null-vector of Liouvillian
\begin{equation}
\LL \rho_{\infty} = 0
\end{equation}
and which shall be  termed {\em non-equilibrium steady state} (NESS). This state is relevant in the asymptotic long time limit of the evolution of the system, and does not decay despite the dissipative nature of the Lindblad equation. The interesting case, when the fixed point is a pure state $\rho_\infty = \ket{\psi_\infty}\bra{\psi_\infty}$ is of particular appeal in quantum optics and quantum information theory \cite{Diehl}, where such a state $\ket{\psi_\infty}$ is called a {\em dark state}. Sufficient condition for existence of a dark state is that $\ket{\psi_\infty}$ is an eigenstate of $H$ and is annihilated by all jump operators, $L_m \ket{\psi_\infty}=0$. Dark states are also known as examples of decoherence-free states, and are of importance in quantum computing \cite{lidhar}. 

We note that the Lindblad equation (\ref{lind}) can in general describe the system coupled to {\em several} thermal (or chemical, magnetic) baths at different values of the
thermodynamic potentials, hence the resulting asymptotic states $\rho_\infty$ can be considered as intrinsically {\em nonequilibrium}.

Of central importance is the understanding and control of uniqueness of NESS, for example for applications in quantum memories \cite{QM} which could be realized in terms of possibly degenerate (non-unique) NESSs. Another potentially very iteresting application of non-unique (multiple) steady states of Liouvillean quantum dynamics could be in detecting non-equilibrium quantum phase transitions, or regions of criticality \cite{eisertprosen}. Importantly, multiple fixed points of dynamical semigroups are also a crucial part of the concept of decoherence-free subspaces \cite{lidhar}, one of the key ideas for fighting errors in quantum computation.
Precise conditions for the uniqueness of NESS, i.e. independence of $\rho_\infty$ of the initial condition $\rho(0)$,  has been established a while ago \cite{Evans,Spohn}. Essentially, it follows from the Evans theorem \cite{Evans} that NESS is unique if and only if the set of operators $\{H,L_m,m=1,2,\ldots \}$ generates the entire multiplicative operator algebra ${\cal B}({\cal H})$. One the other hand, the situations where
NESS is not unique have been systematically much less explored. For example, in the case of existence of a discrete symmetry, of either (i) Liouvillian $\LL$ as a whole, or (ii) of all members of the set 
$\{H,L_m,m=1,2,\ldots\}$ one might expect non-unique NESSs classified by the eigenvalues of the symmetry operation.

 In this paper we will analyze how the existence of symmetries can allow us to treat the transport problem in a boundary driven open anisotropic Heisenberg XXZ spin 1/2 chain with a micro-canonical constraint, i.e. enforcing the exact conservation of total spin (magnetization). Quite interestingly, our numerical results suggest that the two NESSs corresponding to the zero magnetization sector characterized by the eigenvalue of the parity-like symmetry become indistinguishable in the thermodynamic limit, 
and that in the easy-axis regime (of anisotropy $\Delta > 1$) the spin transport becomes sub-diffusive (insulating in the thermodynamic limit), quite different from the results \cite{pz09,RobinPRE} for the un-constrained boundary-driven XXZ chain. This seems consistent with a recently claimed co-existence of diffusive and insulating transport in a gapped XXZ model \cite{IJS}.
In the Appendix we provide a general mathematical prescription how in both cases (i,ii) of the previous paragraph the symmetry can be facilitated to reduce the problem and block-diagonalize the matrix of the Liouvillian superoperator, whereas 
{\em only} in the second case (ii) the non-uniqueness of NESS is guaranteed by the existence of a symmetry operator and distinct Liouvillian fixed points can be labelled by the eigenvalues of the symmetry operation.

\section{Symmetric boundary driven open XXZ spin chains}

\label{example}

In order to study a simple system which does not have a unique NESS and to provide some additional insight into the debate on quantum transport in one dimension \cite{ZotosReview,IJS}, we shall discuss a finite open anisotropic Heisenberg XXZ spin 1/2 chain on $n$ sites, with the Hamiltonian
\begin{equation}
H=\sum_{i=1}^{n-1} \left [ \sigma^\x_i\sigma^\x_{i+1}+\sigma_i^\y\sigma_{i+1}^\y + \Delta \sigma_i^\z\sigma^z_{i+1} \right ].
\label{Ham}
\end{equation}
$\sigma^{\x,\y,\z}_i$ designate a set of standard Pauli matrices corresponding to physical sites $i=1,2,\ldots,n$ and acting on a Hilbert space ${\cal H}=(\CC^2)^{\otimes n}$ of dimension $N=2^n$.
We couple the chain to magnetic reservoirs only through the boundary spins $i=1$ and $i=n$. 

Two cases can be considered:
\begin{enumerate}
\item Firstly, we take a set of four Lindblad operators
\begin{eqnarray} 
L^{\rm w}_1 &=& \sqrt{\Gamma(1-\mu)} \sigma^+_1, \qquad L^{\rm w}_3 = \sqrt{\Gamma(1+\mu)}  \sigma^+_n,  \label{XXZlindweak} \\
L^{\rm w}_2 &=& \sqrt{\Gamma(1+\mu)} \sigma^-_1, \qquad L^{\rm w}_4 = \sqrt{\Gamma(1-\mu)} \sigma^-_n . \nonumber
\end{eqnarray}
where two real parameters $\Gamma > 0$ and $\mu \in [0,1]$ designate, respectively, the strength of coupling to a pair of magnetic reservoirs and strength of (non-equilibrium) driving.
$\sigma^\pm_i = (\sigma^\x_i \pm \ii \sigma^\y_i)/2$ are the spin-flip operators.
\item Secondly, we consider a pair of Lindblad operators
\begin{equation}
L^{\rm s}_1 = \Gamma (1-\mu)\sigma^+_1 \sigma^-_n, \qquad
L^{\rm s}_2 = \Gamma (1+\mu)\sigma^-_1 \sigma^+_n, \label{XXZlindstrong}
\end{equation}
again parametrized by two coupling parameters $\Gamma,\mu$ having the same interpretation as above.
We note that the Liouvillian flow (\ref{lind}) equipped with (\ref{XXZlindstrong}) strictly conserves the total magnetization $M = \sum_{i=1}^n \sigma^\z_i$, at the expense of non-local coupling between the first and the last site. In fact, we can imagine Lindblad operators here as incoherent quantum jumps which transfer spin excitations between $i=1$ and $i=n$ sites while conserving total number of spin-excitations (quasi-particles). Alternatively, such a model can be interpreted as a spin ring where hopping is fully coherent on bonds $(1,2),(2,3)\ldots,(n-1,n)$ while it is fully dissipative (and asymmetric for $\mu\neq 0$) on one bond $(n,1)$, and is a particular case
of quantum exclusion process complementing the one presented in Ref.~\cite{Eisler} (see also \cite{Temme,markoXXd}).
\end{enumerate}

We note that the quantum transport models described by the Hamiltonian (\ref{Ham}) and (\ref{XXZlindweak}), or (\ref{XXZlindstrong}), can be considered as out of equilibrium steady state models with minimal incoherent input, namely with incoherent processes taking place only on the boundary, or on the single bond.

Let $P$ be a permutation operator which exchanges the site $i$ with $n-i+1$ for all $i$. In fact, $P$ can be written as a unitary operator over ${\cal H}$ with explicit representation in the computational basis 
(eigenbasis of $\sigma^\z_i$, $\sigma^\z_i \ket{ a_1,a_2 \ldots a_n} = (1-2 a_i) \ket{ a_1,a_2 \ldots a_n}$ for $a_i \in \{0,1\}$) as
\begin{equation}
P = \sum_{(a_1,a_2,\ldots,a_n)\in \{0,1\}^n} \ket{a_1,a_2,\ldots,a_n}\bra{a_n,a_{n-1},\ldots,a_1}
\end{equation}
and can be interpreted as a {\em parity operator}.
Combining it with spin flips on all sites we define another unitary (parity-like) operator $S$ as
\begin{equation}
S = P \prod_{i=1}^n \sigma^\x_i.
\end{equation}
Clearly, $S$ is a $\ZZ_2$ symmetry of the Hamiltonian, as $[H,S]=0$ and $S^2=1$.

In case (i) with (\ref{XXZlindweak}) one notes that $S L^{\rm w}_1 S^\dagger \equiv \hat{S}  L^{\rm w}_1 = L^{\rm w}_4$ and $\hat{S} L^{\rm w}_2 = L^{\rm w}_3$ hence the Liouvillian flow (\ref{lind}) exhibits a {\em weak} symmetry in the sense of Appendix, Eq. (\ref{weak}). Precisely the same case of symmetrically boundary driven XXZ chain (\ref{Ham},\ref{XXZlindweak}) has been discussed in several recent papers \cite{itaslo,itaslo2,pz09,PZ2010,ProsenPRL,markoPRL} and has been shown to exhibit interesting
non-equilibrium trasport properties. In this case, one can use Evans theorem \cite{Evans} in order to prove that NESS is unique.
The number of invariant subspaces is $n_S = 2$, with symmetry $S$ eigenvalues $s_1 = +1$ and $s_2 = -1$. Let the complete set of 4 symmetry adapted operator sub-spaces be now suggestively denoted as
$\{ s_\alpha,s_\beta\} \equiv {\cal B}_{\alpha,\beta}$. As pointed out in the Appendix (also explaining the notation), the symmetry $\hat{S}$ can still be used in
order to reduce the Liouvillian to two blocks ${\cal B}_1 = \{+1,+1\} \oplus \{-1,-1\}$, and ${\cal B}_2 = \{+1,-1\} \oplus \{-1,+1\}$, which have roughly similar dimensions for large $n$.

Let us now focus on case (ii). $S$ now becomes a {\em strong} symmetry (Appendix, Eq. (\ref{strong})) of the flow (\ref{lind},\ref{XXZlindstrong}) as $\hat{S} L^{\rm s}_m = L^{\rm s}_m$ for $m=1,2$. In this case, however, we have an additional
strong (continuous, $U(1)$) symmetry, generated by magnetization operator $S_M = e^{\ii \varphi M}$, as also $\hat{S}_M L^{\rm s}_m = L^{\rm s}_m$ and $\hat{S}_M H = H$.
We shall say that the flow (\ref{lind},\ref{XXZlindstrong}) obeys a {\em micro-canonical constraint}.
Thus we can use the Theorem \ref{theorem} of Appendix in order to prove the existence of a pair of distinct NESSs for each eigenvalue of total magnetization $s_{\z} \in \{-n,-n+2,\ldots,n-2,n\}$, i.e. we have a $2 N+1$ dimensional 
convex set of fixed points of the Liouvillian flow.\footnote{The dimension of the convex set of NESS is by one smaller than the number of distinct symmetry labelled fixed points $2 (N-1)$ due to the trace constraint $\tr \rho=1$.} As from the quantum transport point of view the most interesting is the zero magnetization sector, we will in the following fix $s_{\z}=0$ and assume the size $n$ to be even.
There we have a line of degenerate NESS parametrized by a number from a unit interval $u \in [0,1]$
\begin{equation}
\rho_{\infty}(u) = u \rho^{1}_{\infty} + (1-u)\rho^{2}_{\infty}, \quad
\rho^{1}_{\infty} \in \{ +1,+1\},\; \rho^{2}_{\infty} \in \{ -1,-1\}.
\end{equation}

\subsection{Numerical results and transport properties}

Let us now turn to the physical properties of our microcanonically constrained open XXZ chain, in particular to the steady state spin current and magnetization profiles. 
The spin current operator corresponding to the bond $(i,i+1)$ is defined as, 
\begin{equation}
j_i= \sigma^\x_i\sigma^\y_{i+1}-\sigma_i^\y\sigma_{i+1}^\x = 2\ii (\sigma^-_i \sigma^+_{i+1}-\sigma^+_i \sigma^-_{i+1})
\end{equation}
and obeys the operator continuity equation of the form $\frac{\dd}{\dd t} (\sigma^\z_j/2) = \ii [H,\sigma^\z/2] = j_i - j_{i-1}$, for $i=2,3\ldots,n-1$. Steady state expectation values shall be denoted as $J=\ave{j_i}$ (which does not depend on site index $i$ due to continuity equation), and $M_i = \ave{\sigma^\z_i}$, where $\ave{\bullet} \equiv \tr (\bullet ) \rho_\infty$.  
We have checked that the fixed point problem for Liouvillian flow (\ref{lind},\ref{XXZlindstrong}) does not admit a closed form solution of a matrix product operator form of small finite rank as in the weakly-symmetric case \cite{ProsenPRL}, thus we had to resort to numerical simulations. Of course, we optimized our calculations by restricting the Liouvillean to diagonal blocks $\LL|_{\{+1,+1\}}$ (for $\rho^1_{\infty}$) and $\LL|_{\{-1,-1\}}$ (for $\rho^{2}_{\infty}$). For $n=4,6,8$ and $10$ sites we used exact numerical diagonalization, while for $n=12,14$ and $16$ we used a wave-function Monte-Carlo approach called the method of quantum trajectories, as outlined in Appendix A of Ref.~\cite{itaslo2} (see also Ref.~\cite{Michel:08}). An efficient Trotter expansion of the propagator $e^{\ii H \delta_t}$ with complex coefficients \cite{Pizorn} has been employed. 
The stochastic time-dependent Shr\" odinger equation has been simulated until the current was equal on all bonds, and the statistical error was small, both within the accuracy better than $1\%$.

\begin{figure}[!]
\begin{center}
\includegraphics[width=0.9\textwidth]{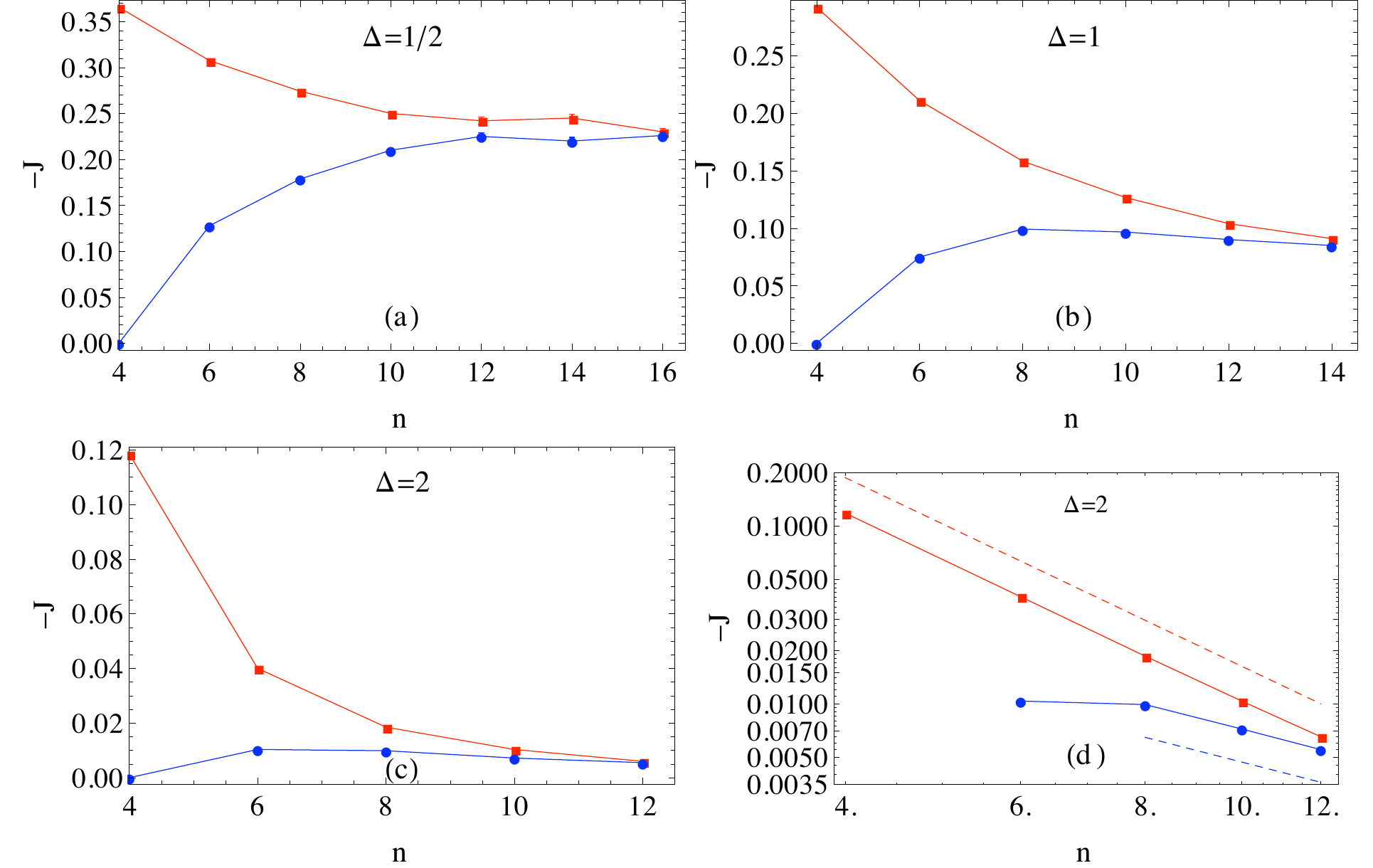}
\end{center}
\caption{We plot negative current $-J$ against spin chain length $n$ for $\Delta=1/2$ (a), $\Delta=1$ (b), and $\Delta=2$ (c), and for NESSs from both symmetry subspaces,
$\{+1,+1\}$ (red/squares) and $\{-1,-1\}$ (blue/circles). Data for $n\le 10$ are numerically exact, while for $n\ge 12$ the statistical/Trotter error is always smaller than the size of the symbols.
In panel (d) we plot the case $\Delta=2$ in the log-log scale, suggesting the fast decay of $J(n)$. The two eye-guiding dashed lines have slopes $-2.68$ (red) and $-1.44$ (blue).}
\label{fig:cur}
\end{figure}

We consider three characteristic values of the anisotropy parameter: $\Delta=1/2$, $\Delta=1$, and $\Delta=2$, where the linear-response-transport in the grand-canonical ensemble (or in an open chain without the microcanonical constraint) has been found to be, respectively, ballistic \cite{ProsenPRL}, anomalous (super-diffusive/sub-ballistic) \cite{markoPRL}, and diffusive \cite{Fabian,RobinPRE}.
On the other hand, linear-response transport in the micro-canonical ensemble in the latter case ($\Delta=2$) has been suggested to be sub-diffusive (insulating) \cite{Prelovsek}.
We chose the bath-coupling parameters as $\Gamma=1$, $\mu=0.2$, which put our setup in a near-linear-response regime, however, we also tried other values of bath parameters and the results did not change qualitatively. By considering a large driving parameter close to the maximum $\mu\approx 1$, we have found {\em negative differential conductance} for both symmetry subspaces (for $n=8$), 
consistent with the behavior found in the unconstrained model \cite{itaslo,itaslo2}.

The spin chain with four sites is particularly interesting as the subspace $\{-1,-1\}$ is one-dimensional, i.e, the NESS in this subspace has to be pure state, i.e. a {\em dark state}.
 Being a pure state it cannot support current ($J=0$) and has a vanishing magnetization on all 4 sites.
 Explicitly, the dark state then is of the form
\begin{equation}
\ket{\psi_\infty}=\frac{1}{\sqrt{2}} \left ( \ket{0110}-\ket{1001} \right )
\end{equation}
and does not depend on any system's parameters, $\Delta,\Gamma,\mu$.
We did not find any dark states of our flow (\ref{lind},\ref{XXZlindstrong}) for $n \ge 6$.

\begin{figure}[!]
\begin{center}
\includegraphics[width=1.03\textwidth]{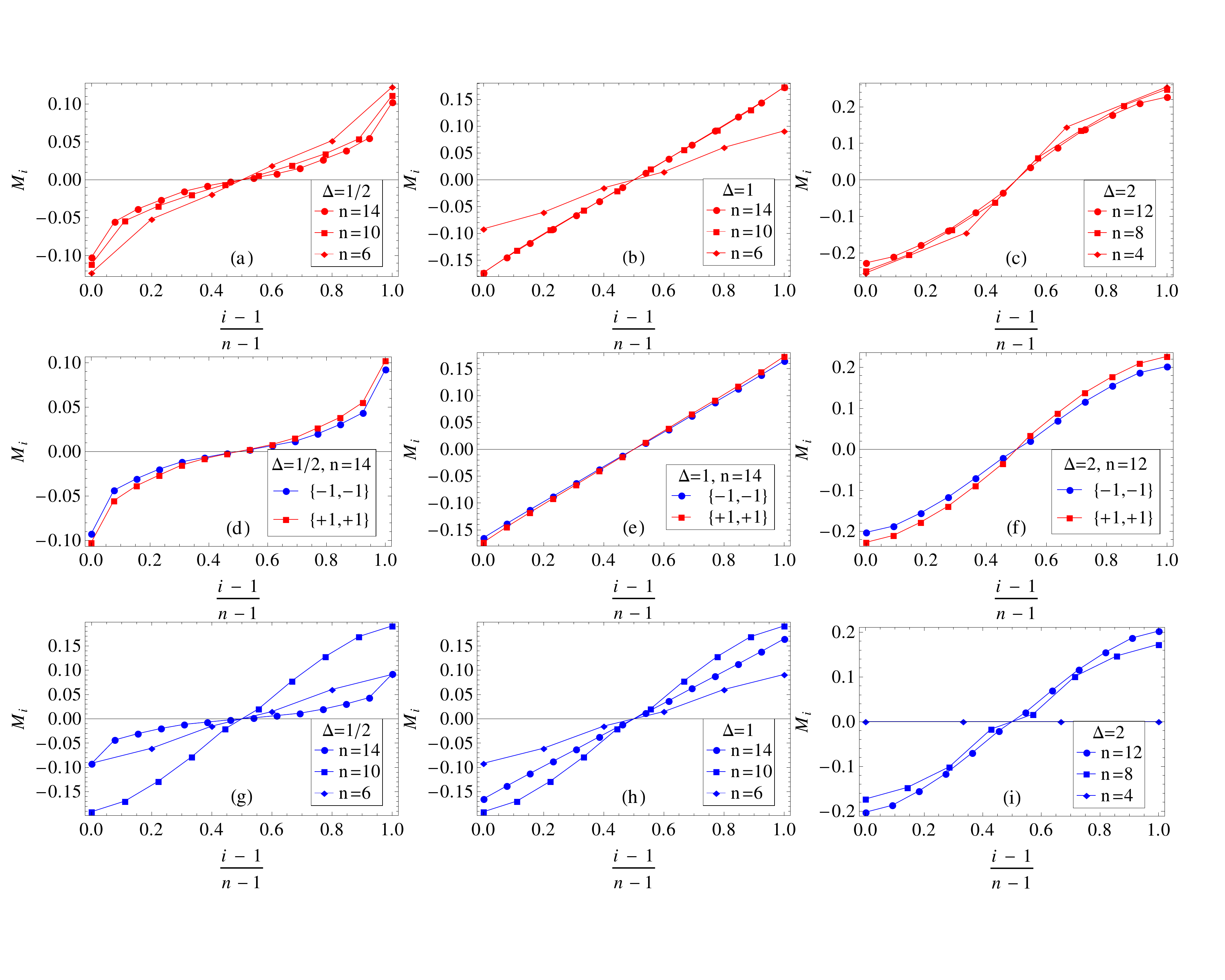}
\end{center}
\vspace{-10mm}
\caption{The magnetization profiles for various values of $\Delta=1/2,1,2$ in the $\{+1,+1\}$ subspace, (a,b,c), respectively, and comparing different chain sizes (see the legends inside plots), and similar for the 
$\{-1,-1\}$ subspace in the lower row of plots (g,h,i). In the midle row of plots (d,e,f) we compare, again for $\Delta=1/2,1,2$, and for the largest system size attainable, the magnetization profiles from the two symmetry sectors (red/blue).}
\label{fig:mag}
\end{figure}

In Fig.~\ref{fig:cur} we display the behavior of the steady state current $J$ versus the chain length $n$ for the two distinct steady states $\rho^{(1)}_\infty$ and $\rho^2_{\infty}$ in the spaces
$\{+1,+1\}$ and $\{-1,-1\}$. We find a general trend of convergence of the two curves which suggest that the current might not depend on the symmetry sector in the thermodynamic limit and would therefore be unique.
The same applies to magnetization profiles displayed in Fig.~\ref{fig:mag}. This implies a kind of {\em ergodicity} and suggests physical irrelevance of the symmetry $S$ in the thermodynamic limit.
Furthermore, in the regime $\Delta=1/2$ the current $J(n)$ tends to converge to a constant suggesting a ballistic behavior, for $\Delta=1$ $J(n)$ seems to slowly (perhaps slightly super-diffusively) decrease with $n$, while for
$\Delta=2$, $J(n)$ decreases fast. As far as our numerics can suggest it seems that in the regime $\Delta=2$, the decrease of $J(n)$ is definitely much faster than $1/n$ which is compatible with the suggested insulating behavior \cite{Prelovsek}. In Fig.~\ref{fig:mag} we make a detailed analysis of the magnetization profiles in both symmetry sectors, and for different system sizes. We find, consistently, that the profiles in the regimes $\Delta=1/2$, $\Delta=1$, and $\Delta=2$, display ballistic, (slightly super-)diffusive, and sub-diffusive (insulating) behavior, respectively. In is perhaps interesting to remark that in the isotropic case $\Delta=1$, for the largest system size that we could simulate, the profile seems to be almost perfectly linear which might also be compatible with thermodynamically normal diffusive behavior in this regime (perhaps related to \cite{RobinPRL}).

\begin{figure}[!]
\begin{center}
\leavevmode
\includegraphics[width=\textwidth]{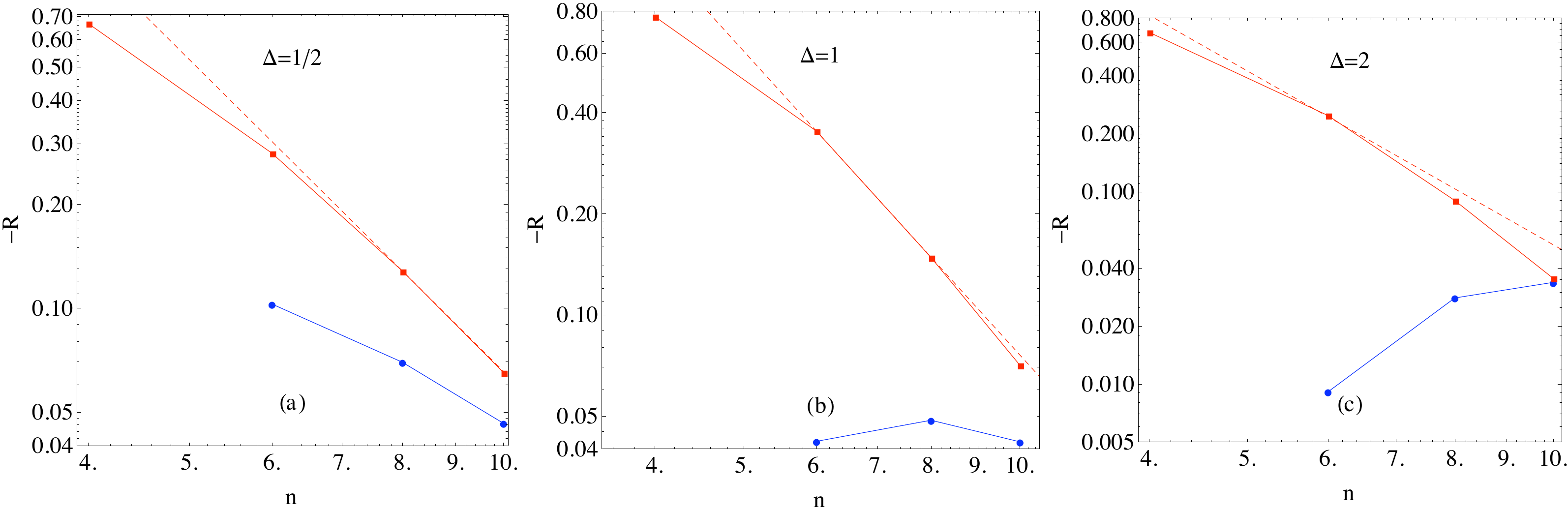}
\end{center}
\caption{
The scaling of the spectal gap $R$ with the system size $n$ in the log-log scale, for the three values of anisotropy $\Delta=1/2$ (a), $\Delta=1$ (b), and $\Delta=2$ (c), and for the two symmetry subspaces,
$\{+1,+1\}$ (red/squares) and $\{-1,-1\}$ (blue/circles). The dashed lines indicate the slope $-3$.}
\label{fig:spec}
\end{figure}

In Fig.~\ref{fig:spec} we focus on the dependence of the relaxation rates of $\rho^\alpha(t)$ to the respective steady state $\rho^\alpha_\infty$ on the chain length $n$. 
The relaxation rate is determined as the Liouvillean spectral gap $R = \min_{\lambda_j \ne 0} (-{\rm Re}\,\lambda_j)$ where $\lambda_j$ denote the eigenvalues of the Liouvillian superoperator $\LL$ in the respective symmetry subspace ($\{+1,+1\}$ or $\{-1,-1\}$). We find consistently that in the $\{+1,+1\}$ sector the spectral gap always quickly decays with $n$, perhaps faster than $n^{-3}$, which is the gap-scaling derived analytically or numerically suggested for some integrable spin chains with boundary driving \cite{NJP:08,markoPRL}. On the other hand, we are unable to conclude anything meaningful on the scaling of the gap in the $\{-1,-1\}$ symmetry sector, where the behavior of $R(n)$ may not even be monotonic.
To complement the spectral information we plot in Fig.~\ref{fig:fullspec} the set of all Liouvillian eigenvalues $\lambda_j$ which lie sufficiently close to the imaginary line (i.e. with sufficiently small damping rates
$-{\rm Re}\lambda_j$). Further, we considered the density distribution of Liouvillian spectrum projected onto the real line, in fact a cumulative distribution
$W(r) = \sum_{\lambda_j} \theta(r - {\rm Re}\,\lambda_j)/\sum_{\lambda_j}1$ giving the probability that a randomly picked damping rate is larger than $-r$. Here $\theta(r)$ designates a Heaviside step function.
The plot of $W(r)$ for the two symmetry sectors shown in Fig.~\ref{fig:distspec} in the easy-axis regime $\Delta=2$ suggests existence of two pseudo-gaps, particularly clearly for the $\{-1,-1\}$ symmetry sector.

\begin{figure}[!]
\begin{center}
\leavevmode
\includegraphics[width=\textwidth]{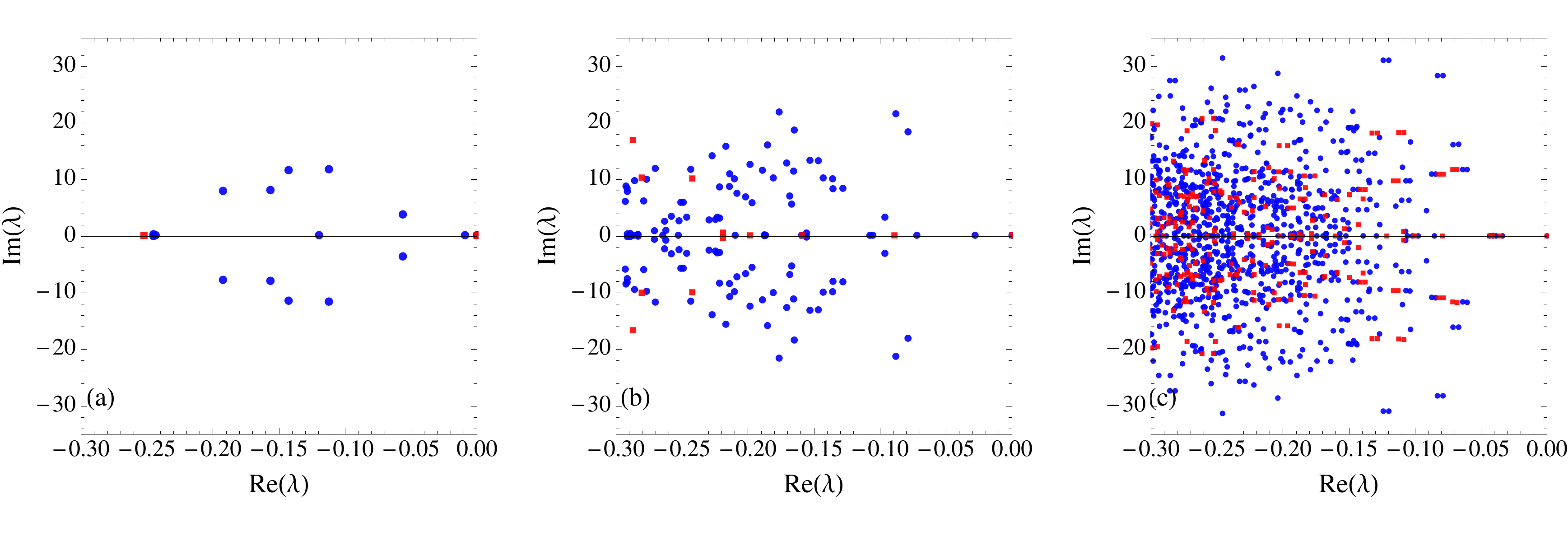}
\end{center}
\caption{The full Liouvillian spectra $\{\lambda_j\}$ for the case of $\Delta=2$ and for three different chain sizes $n=6$ (a), $n=8$ (b), and $n=10$ (c), and for the two symmetry subspaces,
$\{+1,+1\}$ (red/squares) and $\{-1,-1\}$ (blue/circles). We plot just a small part of the complex plane close to the imaginary line (i.e. plotting just the smallest rates corresponding to longest-lived modes).}
\label{fig:fullspec}
\end{figure}

\begin{figure}[!]
\begin{center}
\leavevmode
\includegraphics[width=0.5\textwidth]{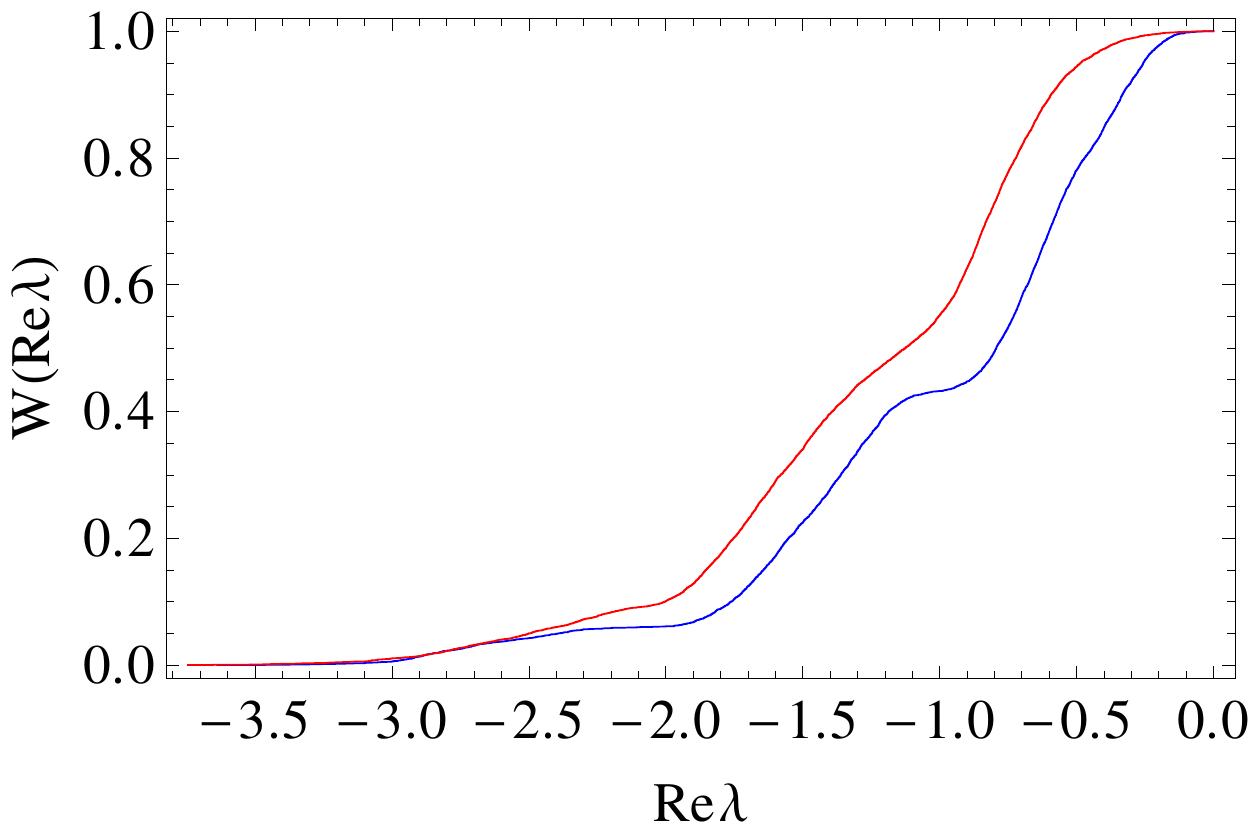}
\end{center}
\caption{The cumulative distribution $W(r)$ of the the case $\Delta=2$, $n=10$, i.e. the probability that the randomly chosen decay rate $-{\rm Re}\,\lambda_j$ is larger than $-r$,
for the two symmetry subspaces, $\{+1,+1\}$ (red) and $\{-1,-1\}$ (blue). Note the occurrence of two pseudo-gaps (strips in complex plane with very few -- rarely distributed decay modes), 
particularly visible in the data for  $\{-1,-1\}$.}
\label{fig:distspec}
\end{figure}

\section{Discussion and conclusion}

We have presented a small but potentially very useful observation on the symmetry reduction of the Linbdblad master equations describing Markovian open quantum systems.
We classified quantum Liouvillian symmetries as {\em weak} or {\em strong}, depending, respectively, on whether only the generator of master equation as a whole commutes with the symmetry operation, or whether
each of the Lindblad (jump) operators, and the Hamiltonian, commute with the symmetry individually. We have shown that only existence of a strong symmetry implies nontrivial symmetry reductions inside the space of steady states (fixed points of the Markovian semi-group). We have also provided a non-trivial example of our construction in the context of quantum transport problem for the strongly interacting XXZ spin 1/2 chain.
Namely defining an open XXZ chain with a micro-canonical constraint, we have shown that it exhibits a strong parity-type symmetry and outlined numerical computation of distinct steady states and the corresponding physical observables. Our numerical results seem to suggest that microcanonically constrained open XXZ chain exhibits a sub-diffusive (thermodynamically insulating) behavior in the gapped Ising-like (easy axis) phase.

We note that both, weak and strong symmetries can provide practical advantage in description of quantum Liouvillian flows, as they both allow to block-diagonalize and hence reduce the dimension of the Liouvillian.
One can also have a combination of both, namely independent generators of weak and strong symmetries, or non-abelian symmetries of either type. In particular, the study of non-abelian Liouvillian symmetries should be an important line of future research.

\section*{Acknowledgements}
TP acknowledges useful discussions with Wojciech de Roeck and Herbert Spohn. The work has been financially supported by the grants P1-0044 and J1-2208 of the Slovenian research agency (ARRS).

\section{Appendix: Symmetry reductions of quantum Liouvillian dynamics}

\label{decomposition}

Let us consider two cases of {\em symmetric} Lindblad master equations with two different kinds of dissipators of the flow (\ref{lind}):
\begin{enumerate}
\item 
For the first kind, which we will call a {\em weak symmetry}, we assume that there exist a unitary operator $S$ over ${\cal H}$, such that 
\begin{equation}
\LL (S x S^\dagger) = S (\LL x) S^\dagger
\label{weak}
\end{equation}
for any $x\in {\cal B}({\cal H})$.
\item 
For the second kind, which we will call  a {\em strong symmetry}, we assume that a unitary operator $S$ over ${\cal H}$ commutes with each member of the set $\{H,L_m,m=1,2,\ldots\}$ individually
\begin{equation}
[S,H] = 0, \quad [S,L_m] = 0,\; m=1,2,\ldots
\label{strong}
\end{equation}
\end{enumerate}
Clearly (ii)  implies (i), but not vice versa.

We continue by fixing some notation. Let $s_\alpha = e^{\ii \theta_\alpha}$, $\alpha=1,2\ldots n_S$ denote distinct eigenvalues of $S$, and let ${\cal H}_\alpha$ be the corresponding mutually orthogonal eigenspaces, so that we have a complete symmetry decomposition of the Hilbert space
\begin{equation}
{\cal H} = \bigoplus_{\alpha=1}^{n_S} {\cal H}_\alpha .
\end{equation}
The adjoint representation of $S$ on Hilbert space ${\cal B}({\cal H})$ shall be denoted as ${\hat S}$,
\begin{equation}
\hat{S}(x) = S x S^\dagger .
\end{equation}
The spectrum of a super-operator ${\hat S}$ is formed of all possible products $s_\alpha \bar{s}_\beta = e^{\ii (\theta_\alpha-\theta_\beta)}$,
since  ${\hat S}(\ket{\psi}\bra{\phi}) = s_\alpha \bar{s}_\beta \ket{\psi}\bra{\phi}$ for any $\ket{\psi}\in{\cal H}_\alpha$,
$\ket{\phi}\in{\cal H}_\beta$. Let $s'_\nu$, $\nu=1,2,\ldots n_{\hat S}$ denote the distinct eigenvalues of 
$\hat{S}$, denoted such that $s'_1 = 1$, and $s'_\nu \neq 1$ for all $\nu \neq 1$. Note that the following bounds always hold
\begin{equation}
n_{S} \le n_{\hat S} \le n_S (n_S - 1).
\end{equation}
The lower bound $n_{\hat S} = n_S$ is reached when $\{ s_\alpha \}$ are roots of unity, i.e. in the case of $\ZZ_{n_S}$ symmetry obeying $S^{n_S} = 1$. The upper bound 
$n_{\hat S} = n_S (n_S - 1)$ is reached when there is no arithmetic structure in the spectrum $\{ s_\alpha \}$ which may happen when $S$ generates a continuous group of symmetry transformations.
Let 
\begin{equation}
{\cal B}({\cal H}) = \bigoplus_{\nu=1}^{n_{\hat S}} {\cal B}_{\nu}
\label{weakdecomp}
\end{equation}
be the corresponding symmetry decomposition of the operator space
${\cal B}({\cal H})$, where the $\hat{S}$-eigenspaces ${\cal B}_{\nu}$ are mutually orthogonal in the Hilbert-Schmidt sense.

Clearly, having a weak symmetry (\ref{weak}) is equivalent to having commuting superoperators
\begin{equation}
\LL {\hat S} = {\hat S} \LL
\end{equation}
which means that the Liouvillian is block diagonalized with respect to the decomposition (\ref{weakdecomp}), namely
\begin{equation}
\LL {\cal B}_\nu \subseteq {\cal B}_\nu,
\label{Bnu}
\end{equation}
i.e. all spaces ${\cal B}_\nu$ are invariant w.r.t. the flow (\ref{lind}).
Note that only ${\cal B}_1$ contains operators with non-vanishing trace, thus it should contain the fixed point $\rho_\infty \in {\cal B}_1$ due to the trace preservation.
We remark that under the conditions of the Evans theorem \cite{Evans} NESS can be unique despite the existence of a weak symmetry, as shall be the case in the example outlined in the next section.

Let us now assume the existence of a strong symmetry. In this case we can prove the following useful result:

\begin{theorem}
\label{theorem}
Existence of a unitary operartor $S$ with the property (\ref{strong}) implies:
\begin{enumerate} 
\item Liouvillian $\LL$ can be block-decomposed into $n_S^2$ invariant subspaces
$$\LL {\cal B}_{\alpha,\beta} \subseteq {\cal B}_{\alpha,\beta},$$ where
${\cal B}_{\alpha,\beta} = \{ \ket{\psi}\bra{\phi}; \ket{\psi} \in {\cal H}_\alpha, \ket{\phi} \in {\cal H}_\beta \}$ for $\alpha,\beta=1,2,\ldots n_S$.
\item We have at least $n_S$ distinct fixed points (NESSs), which we can label by $s_\alpha$, namely each diagonal operator space
contains at least one fixed point, $$\rho_\infty^{\alpha} \in {\cal B}_{\alpha,\alpha}$$ for $\alpha=1,2,\ldots n_S$.
\end{enumerate}
\end{theorem}
\begin{proof}
Let us write the super-operators of left and right multiplications with the symmetry operator $\hat{S}_{\rm L}(x) \equiv S x$, $\hat{S}_{\rm R}(x) \equiv x S^\dagger$, for any $x \in {\cal B}({\cal H})$.
Since, clearly, $[\hat{S}_{\rm L},\hat{S}_{\rm R}] = 0$, the spaces ${\cal B}_{\alpha,\beta}$ can be considered as the joint eigenspaces of both $\hat{S}_{\rm L}$ and $\hat{S}_{\rm R}$. Namely $\hat{S}_{\rm L}(\ket{\psi}\bra{\phi}) = s_\alpha \ket{\psi}\bra{\phi}$,   $\hat{S}_{\rm R}(\ket{\psi}\bra{\phi}) = \bar{s}_\beta \ket{\psi}\bra{\phi}$, for any 
$\ket{\psi}\in {\cal H}_\alpha$, $\ket{\phi}\in {\cal H}_\beta$.
Eq. (\ref{strong}) implies commutativity $[\LL,\hat{S}_{\rm L}] = 0$, $[\LL,\hat{S}_{\rm R}] = 0$, implying that any ${\cal B}_{\alpha,\beta}$ is also invariant subspace of $\LL$ proving point (i) of the theorem.

To prove (ii) we note again that operators with non-vanishing trace can only be contained in diagonal spaces ${\cal B}_{\alpha,\alpha}$, due to mutual orthogonality of ${\cal H}_\alpha$.
Hence starting from some $\rho^\alpha(0) \in {\cal B}_{\alpha,\alpha}$ with $\tr \rho^\alpha(0) = 1$ (say $\rho(0) = P_\alpha /\tr P_\alpha$ where $P_\alpha$ is orthogonal projector to the subspace ${\cal H}_\alpha$),
the flow (\ref{lind}) yields a fixed point $\rho^\alpha_\infty$ in ${\cal B}_{\alpha,\alpha}$ for any $\alpha=1,\ldots,n_S$. Again according to Evans theorem \cite{Evans} this fixed point can be unique, or further degenerate. In any case, we have at least $n_S$ distinct steady states, which may be labelled by the symmetry eigenvalue $s_\alpha$. QED
\end{proof}

Clearly, the direct sum of all invariant subspaces give the entire operator space
\begin{equation}
{\cal B}({\cal H}) = \bigoplus_{\alpha=1}^{n_S}\bigoplus_{\beta=1}^{n_S} {\cal B}_{\alpha,\beta},
\label{decomp1}
\end{equation}
and the quotient invariant spaces (\ref{weakdecomp}) are given by partial direct sums
\begin{equation}
{\cal B}_\nu = \bigoplus_{\alpha,\beta}^{s_\alpha \bar{s}_\beta = s'_\nu} {\cal B}_{\alpha,\beta} .
\label{decomp2}
\end{equation}

We remark that (non-positive) initial conditions $\rho(0)$ supported in the non-diagonal invariant spaces ${\cal B}_{\alpha,\beta}$, $\alpha\neq \beta$, which have vanishing trace, could in principle provide extra degeneracy of the null space of $\LL$, i.e. an extra dimension of the convex set of NESSs, or simply complete the description of the relaxation process for a general initial condition $\rho(0)$. In other words, we do not see any argument which would forbid the spectra of  block-operators $\LL|_{{\cal B}_{\alpha,\beta}}$ to contain $0$ even for $\alpha \ne \beta$, although this should be -- if possible at all -- an exceptional (non-generic) situation.

\end{document}